\documentclass[11pt]{article}
\usepackage{style}
\usepackage{physics}
\usepackage[margin=1in]{geometry}
\usepackage{amsmath,amssymb,amsthm,mathtools}
\usepackage{bbm}
\usepackage{enumitem}

\newcommand{\Sep}{\operatorname{Sep}}

\newcommand{\Frob}{\operatorname{F}}
\newcommand{\ot}{\otimes}
\renewcommand{\ip}[2]{\langle{#1},{#2}\rangle}

\title{Quantum Separability in Polynomial Time}
\author{Giulio Malavolta\thanks{Bocconi University \href{mailto:giulio.malavolta@unibocconi.it}{giulio.malavolta@unibocconi.it}}}
\date{}

\begin{document}

\maketitle

\begin{abstract}
The quantum separability problem asks whether a bipartite density matrix is separable or is $\eta$-far from every separable state. We give a randomized polynomial-time algorithm for this problem for every fixed constant gap $\eta>0$, when distance is measured in the Euclidean norm.
\end{abstract}

\section{Introduction}

 Entanglement is the fundamental property distinguishing quantum information from classical probability, and it underpins uniquely quantum phenomena such as violation of Bell inequalities. A central question in computer science and physics is to characterize entanglement in quantum states.

A bipartite density matrix $\rho\in L(\mathbb C^d\otimes \mathbb C^d)$ is \emph{separable} if it lies in the convex hull of pure product states $\Sep(d,d)$,
and is otherwise \emph{entangled}. Deciding whether a given density matrix lies in $\Sep(d,d)$ is the \emph{quantum separability problem}. Peres \cite{Peres1996Separability} showed that separability implies positivity under partial transposition, and Horodecki, Horodecki, and Horodecki \cite{Horodecki1996Separability} proved that this criterion is also sufficient in $2\times 2$ and $2\times 3$ dimensions, but not in general. This initiated a broad study of separability criteria \cite{Donald2002Uniqueness,Bennett1996MixedState,Vedral1997Quantifying,Vidal2002Computable}.

However, strong barriers are known for a \emph{computationally efficient} criteria: Gurvits \cite{Gurvits2003Edmonds} proved that \emph{weak membership} for the set of separable states is NP-hard. That is, he showed NP-hardness of determining membership in $\Sep(d,d)$ with accuracy $\eta = \exp(-d)$, in Euclidean distance. This was later improved by Gharibian \cite{Gharibian2010StrongNPHardness} to \emph{strong} NP-hardness, with accuracy $\eta = \frac{1}{\poly(d)}$.  Harrow and Montanaro \cite{Harrow2013TestingProductStates} showed evidence of computational hardness even in the \emph{constant} error regime for separability-type problems, but in the trace-norm scale.

On the algorithmic side, the semidefinite programming hierarchy of Doherty, Parrilo, and Spedalieri \cite{Doherty2004CompleteFamily} gives a complete family of separability criteria by searching for symmetric extensions. A breakthrough result of Brandão, Christandl, and Yard \cite{Brandao2011Quasipolynomial} proves quantitative convergence bounds for such a hierarchy, resulting into a quasi-polynomial algorithm for the weak membership problem, for \emph{constant} accuracy $\eta$. Their result applies both the Euclidean norm and the LOCC norm, a relevant quantity in quantum information. Subsequently, Shi and Wu \cite{Shi2011EpsilonNet} achieve similar results with a different algorithmic approach, based on $\varepsilon$-nets.

A long-standing open problem is whether there exists a polynomial-time algorithm for quantum state separability, in the constant-accuracy regime.

\subsection{Our Results}

This paper resolves the constant-gap Euclidean weak-membership problem. For every fixed constant $\eta>0$, we give a randomized algorithm running in time polynomial in $d$ that distinguishes the cases
\[
\rho \in \Sep(d,d) \quad \text{and} \quad \min_{\psi \in \Sep(d,d)}\norm{\rho - \psi}_{\Frob} > \eta
\]
where $\norm{\cdot}_{\Frob}$ is the Euclidean/Frobenius norm. Besides the weak membership problem, we record some other useful consequences of this algorithm.

The \emph{Best Separable State} (BSS) problem asks, given a Hermitian operator $M\in L(\mathbb C^d\otimes \mathbb C^d)$, to approximate the support function of the separable states $h_{\Sep}(M)$ -- see \Cref{eq:hSep-intro}. Our algorithm allows us to approximate $h_{\Sep}(M)$ with additive accuracy, for all $M$ with $\norm{M}_{\Frob}\le 1$, thus immediately implying a polynomial-time algorithm for BSS for this special case. The general case $0 \leq M \leq \Id $ is unlikely to admit a polynomial-time algorithm, as shown in \cite{Harrow2013TestingProductStates}.

Similarly, a \emph{mean-field Hamiltonian} with product-pair interaction consists of a Hermitian operator acting on
n sites (each formed by a $d$-dimensional quantum system)
defined as 
\[H_n :=
\frac 2n \sum_{i<j} K_{i,j} 
\]
with $K_{i,j}$ a Hermitian matrix which acts as $K$ on sites $i$ and $j$, acting as the identity on other locations. It is well-known, e.g., \cite{fannes2006finite}, that the computation of the ground-state energy of a mean-field Hamiltonian $H_n$, in the limit $n\to\infty$, can be reduced to computing $-h_{\Sep}(-K)$. Our algorithm can be used to approximate it, up to error $\eta\norm{K}_{\Frob}$. This improves the approach of \cite{Brandao2011Quasipolynomial} from quasipolynomial to polynomial time.

\subsection{Proof Overview}

The main technical step is to design a randomized polynomial-time approximation algorithm for optimizing the support function \begin{equation}\label{eq:hSep-intro}
     h_{\Sep}(M)= \max_{\norm{x}_2=\norm{y}_2=1}(x\otimes y)^*M(x\otimes y)
     \end{equation}
for any Hermitian $M \in L(\C^d\otimes \C^d)$ with $\norm{M}_{\Frob} \leq 1$.
The weak-membership algorithm is then obtained by the standard separation-to-membership reduction, implemented via gradient descent (in \Cref{sec:opt-to-membership}). Our first step is to simplify the problem by decomposing the matrix $M$ into \[ M=\alpha \Id\otimes \Id+A\otimes \Id+\Id\otimes B+C,\] where \(A,B\) are traceless and \(C\) has vanishing partial traces. We can then rewrite
\begin{align*}    
h_{\Sep}(M)&= \max_{\norm{x}_2=\norm{y}_2=1}(x\otimes y)^*M(x\otimes y) = \alpha + \max_{\norm{x}_2=\norm{y}_2=1}x^*Ax + y^*By + (x\otimes y)^*C(x\otimes y)\\
&\leq \alpha + \norm{A}_{\Frob} + \norm{B}_{\Frob} + h_{\Sep}(C) \leq \alpha + \frac{2}{\sqrt{d}} + h_{\Sep}(C)
\end{align*}
where the last inequality holds because $\norm{M}_{\Frob} \leq 1$ and the four summands are Hilbert-Schmidt orthogonal. For a growing $d$, the term $\frac{2}{\sqrt{d}}$ becomes irrelevant, so it suffices to have an algorithm for estimating $h_{\Sep}(C)$, where $C$ satisfies $\Tr_0 C = \Tr_1 C =0$.

The goal will be to express this optimization problem as a constraint satisfaction (CSP) problem. The challenge is to formulate the problem in such a way that the CSP is \emph{dense} and therefore efficiently solvable via known polynomial-time approximation schemes (PTAS) \cite{Arora1995DensePTAS}. For this to be possible, we require the maximizers to be \emph{flat}, i.e., to have entrywise small coefficients.

The key observation is that, after a random change of local basis, Hermitian matrices with vanishing partial traces are entrywise flat in expectation. Namely, sample independent Haar-random unitaries \(U,V \sim \mathrm{Haar}(\mathcal{U}(d))\) and replace \(C\) by \[ \hat C=(U\otimes V)^*C(U\otimes V).\] 
Since local unitaries preserve product states, \(h_{\Sep}(C)=h_{\Sep}(\hat C)\). Then,  in \Cref{sec:haar-bound}, we show that the entries of $\hat C$ are indeed spread nearly uniformly: We prove that, with probability at least $1-\delta^2$, we can decompose 
\[\hat C = (U\otimes V)^*C(U\otimes V)=C_{\mathrm{flat}}+C_{\mathrm{tail}},\] where \(\norm{C_{\mathrm{tail}}}_{\Frob} =O(\delta)\), while \[\norm{C_{\mathrm{flat}}}_{\rm max} \lesssim \frac{\operatorname{polylog}(1/\delta)}{d^2}.\]
The small tail perturbs \(h_{\Sep}\) by only \(O(\delta)\), so it suffices to optimize the flat instance $C_{\rm flat}$.

Next, in \Cref{sec:disc-main} we prove that flat matrices have near-flat optimizers. Then, to make the problem combinatorial, we prove a dimension-free discretization theorem: For a flat matrix, the continuous product-state optimum is, up to error \(\eta\), equal to the maximum of a suitably defined objective function
over \(x,y\in \mathcal E^d\), where \(\mathcal E\subset \mathbb C\) is an \(\varepsilon\)-net of a disk of constant radius. Crucially, \(|\mathcal E|=\operatorname{poly}(1/\eta,\log(1/\delta))\), independent of the Hilbert-space dimension \(d\). The proof first shows that flat instances have bounded-amplitude near optimizers, and then uses Lipschitz continuity to round each coordinate independently.

It remains to solve the resulting finite-alphabet problem. In \Cref{sec:opt}, the discretized objective is encoded as a dense Max-\(4\)-CSP with variables \[z_i=(x_i,y_i)\in \Omega:=\mathcal E\times \mathcal E,\qquad i\in[d]\]
and the objective function is chosen so that the optimal value coincides with \(h_{\Sep}\).
Since \(|\Omega|\) is independent of \(d\), the dense-CSP PTAS gives, in polynomial time for every fixed accuracy, an assignment whose value is within the required additive error. Mapping this assignment back through the Haar rotation, yields unit vectors \(\hat x,\hat y\) satisfying \[(\hat x\otimes \hat y)^*M(\hat x\otimes \hat y)\approx h_{\Sep}(M).\]

\section{Preliminaries}\label{sec:prelim}

All Hilbert spaces $\mathcal{H}$ are finite-dimensional and we denote by $L(\mathcal H)$ the set of all linear operators acting on $\mathcal{H}$.  We use the Hilbert--Schmidt/Frobenius inner product
\[
 \ip{A}{B}:=\Tr(A^*B),
\]
and write $\norm{A}_{\Frob}^2:=\Tr(A^*A)$ for the Frobenius norm of $A$. For a random $\mathcal{H}$-valued variable $X$, we write 
$\norm{X}_{L_q(\mathcal{H})}:=(\mathbb E\|X\|_{\mathcal{H}}^q)^{1/q}$. When clear from the context, we omit the dependence on $\mathcal{H}$. 

For a matrix $A$, we denote the size of its largest entry by $\norm{A}_{\rm max} = \max_{a,b} \abs{A_{a,b}}$.

\subsection{Hermitian Decomposition}\label{sec:hermite}

For a matrix $M\in L(\C^d\ot\C^d)$, we write its entries as
\[
 M_{ac,be}:= (u_a \otimes u_c)^* M (u_b\ot u_e)
\]
where $a,b,c,e\in[d]$ and $u_j$ denotes the $j$-th standard basis vector of $\C^d$. We write
\(\Tr_0\) and \(\Tr_1\) for the partial traces over
the first and second tensor factors, i.e.,
\[
(\Tr_0 M)_{c,e}=\sum_a M_{ac,ae},
\quad\text{and}\quad
(\Tr_1 M)_{a,b}=\sum_c M_{ac,bc}.
\]
Every Hermitian $M\in L(\C^d\ot\C^d)$ decomposes uniquely as
\begin{equation}\label{eq:local-decomp}
 M=\alpha \Id\ot\Id+A\ot\Id+\Id\ot B+C,
\end{equation}
where $\Tr A=\Tr B=0$ and $\Tr_0 C=\Tr_1 C=0$.
Explicitly,
\[
 \alpha=\frac{\Tr M}{d^2}\quad\text{and}\quad
 A=\frac1d\Tr_1 M-\alpha \Id
 \quad\text{and}\quad
 B=\frac1d\Tr_0 M-\alpha \Id.
\]  
The four summands in \Cref{eq:local-decomp} are Hilbert--Schmidt orthogonal, and hence
\begin{equation}\label{eq:orthogonal-decomp-norm}
 \norm{M}_{\Frob}^2=|\alpha|^2d^2+d\norm{A}_{\Frob}^2+d\norm{B}_{\Frob}^2+\norm{C}_{\Frob}^2.
\end{equation}
In particular, if $\norm{M}_{\Frob}\le1$, then $\norm{C}_{\Frob}\le1$ and $\norm{A}_{\Frob},\norm{B}_{\Frob}\le d^{-1/2}$.

\subsection{Complex Gaussians}\label{sec:gauss}

A random vector \(X\in\C^d\) is a standard complex Gaussian, denoted
\(X\sim\mathcal N_{\C}(0,\Id_d)\), if
\[
 X=\frac{G+iH}{\sqrt2},
\]
where \(G,H\in\R^d\) are independent real Gaussian vectors with $ G,H\sim\mathcal N(0,\Id_d)$.
We recall the following standard inverse moment bound for the norm of a complex Gaussian vector.
\begin{lemma}\label{lem:inverse-gaussian-norm}
Let \(X\sim\mathcal N_{\C}(0,\Id_d)\) and let $Y:= \norm{X}_2^2$. Then for every integer
\(1\le q\le d/4\),
\[
 \norm{Y^{-1}}_{L_{2q}}
 \le \frac{2}{d}.
\]
\end{lemma}

\begin{proof}
Since \(X\sim\mathcal N_{\C}(0,\Id_d)\), we have that $\|X\|_2^2$ is distributed according to 
\(\Gamma(d,1)\), with density
\[
 f(t)=\frac{1}{(d-1)!}t^{d-1}e^{-t},\qquad t>0.
\]
Thus, for \(2q<d\),
\begin{align*}  
\norm{Y^{-1}}_{L_{2q}}^{2q}=
 \E\|X\|_2^{-4q}
 = \frac{1}{(d-1)!}\int_0^\infty t^{d-1-2q}e^{-t}\,dt
 = \frac{(d-1-2q)!}{(d-1)!}
 \le (d-2q)^{-2q}.
\end{align*}
Taking the \(1/(2q)\)-th power gives
\[
 \norm*{Y^{-1}}_{L_{2q}}
 \le \frac{1}{d-2q}.
\]
If \(q\le d/4\), then \(d-2q\ge d/2\), hence the bound follows.
\end{proof}
We also recall the following consequence of Gaussian hypercontractivity, proven for instance in \cite[Theorem A.1]{adamczak2021moments}.

\begin{lemma}\label{lem:gaussian-hypercontractivity}
Let \(\mathcal H\) be a finite-dimensional Hilbert space, and let
\(X\sim\mathcal N_{\C}(0,\Id_d)\).  Let
\(f:\C^d\to\mathcal H\) be an \(\mathcal H\)-valued polynomial of degree at
most \(m\) in the real and imaginary parts of the coordinates of \(X\).
Then, for every \(r\ge2\),
\[
 \norm{f(X)}_{L_r(\mathcal H)}
 \le
 c_m r^{m/2}\norm{f(X)}_{L_2(\mathcal H)}
\]
for a constant $c_m$ depending only on $m$.
\end{lemma}
Finally, we state the following useful consequence of Wick's formula for complex Gaussians.
\begin{lemma}
\label{lem:complex-gaussian-quadratic-form}
Let \(X\sim\mathcal N_{\C}(0,\Id_d)\) be a standard complex Gaussian vector.
Then, for every \(T\in L(\C^d)\),
\[
 \E \left|
 X^*TX-\Tr T
 \right|^2
 =
 \norm{T}_{\Frob}^2 .
\]
\end{lemma}
\begin{proof}
Write \(X=(X_1,\ldots,X_d)\). Expanding the internal sum,
\[
X^*TX-\Tr T
=
\sum_{i,j=1}^d T_{ij}\bigl({X}^*_iX_j-\delta_{ij}\bigr).
\]
Therefore 
\begin{align*}
\E\left|X^*TX-\Tr T\right|^2
&=
\sum_{i,j,k,\ell=1}^d
T_{ij}T_{k\ell}^*\,
\E
(X^*_iX_j-\delta_{ij})
(X_kX^*_\ell-\delta_{k\ell})\\
&=
\sum_{i,j,k,\ell=1}^d
T_{ij}T_{k\ell}^*\delta_{ik}\delta_{j\ell}
=
\sum_{i,j=1}^d |T_{ij}|^2
=
\|T\|_{\Frob}^2
\end{align*}
where the second line uses Wick's complex Gaussian identity $\E X^*_iX_jX_kX^*_\ell
=
\delta_{ij}\delta_{k\ell}
+
\delta_{ik}\delta_{j\ell}
$ together with $\E X_i^*X_j = \E X_iX_j^*=\delta_{ij}$.
\end{proof}

\subsection{The Haar Measure on the Unitary Group}\label{sec:haar}

Let $\mathcal{U}(d)$ be the unitary group in dimension $d$.
The Haar measure on $\mathcal{U}(d)$ is the unique Borel probability measure $\mu$ such that, for every Borel set $A \subseteq \mathcal{U}(d)$ and every $V \in \mathcal{U}(d)$,
\[
\mu(VA)=\mu(AV)=\mu(A).
\]
Equivalently, the Haar measure is the uniform probability measure on $\mathcal{U}(d)$.

It is well known that one can sample from the Haar measure on $\mathcal{U}(d)$ in time polynomial in $d$. For example, one may sample a matrix with independent complex Gaussian entries, compute its QR decomposition, and then normalize the phases of the diagonal entries of the triangular factor. This construction is described, for instance, in \cite{mezzadrirandom}.

To keep the exposition clean, we ignore issues of numerical precision and assume that Haar-random unitaries are sampled with infinite-precision real arithmetic. I.e., the complexity is measured in the number of real-arithmetic operations done by the algorithm, as in prior work \cite{Doherty2004CompleteFamily,Brandao2011Quasipolynomial}. In \Cref{sec:bit} we outline the modifications needed to upgrade the algorithm and the analysis to the finite-precision arithmetic model.

\subsection{Quantum Separability}\label{sec:sep}

Let $\Sep(d,d):=\operatorname{conv}
\left\{
(xx^*)\otimes(yy^*) :
\|x\|_2=\|y\|_2=1
\right\}$ denote the convex hull of bipartite product states on $\C^d\ot\C^d$.  For a Hermitian matrix $M\in L(\C^d\ot\C^d)$, define the separable support function
\[
 h_{\Sep}(M):=\max_{\rho\in\Sep(d,d)}\Tr(M\rho) = \max_{\norm{ x}_2=\norm{ y}_2=1}(x \ot y)^* M(x \ot y).
\]
The \emph{weak membership problem} \cite{Gurvits2003Edmonds,Gharibian2010StrongNPHardness} for $\Sep(d,d)$ with gap $\eta>0$ is a promise problem, whose instances are density matrices $\rho$ in $\C^d\ot\C^d$.
An algorithm for the weak membership problem $\Sep(d,d)$ must satisfy the following conditions:
\begin{itemize}
    \item (Completeness) For all $\rho \in \Sep(d,d)$, the algorithm accepts $\rho$ with probability at least $\frac23$.
    \item (Soundness) For all $\rho$ such that $\min_{\psi\in\Sep(d,d)} \norm{\rho-\psi}_{\Frob} > \eta$, the algorithm accepts $\rho$ with probability at most $\frac13$.
\end{itemize}
One can also consider the asymmetric variant of the problem $\Sep(d_0,d_1)$, where the two Hilbert spaces have different dimensions, but the version stated above with $d_0 = d_1 = d$ is without loss of generality: Let $d=\max\{d_0,d_1\}$, and let $V_0:\C^{d_0}\mapsto\C^d$ and $V_1:\C^{d_1}\mapsto\C^d$ be isometric embeddings. On input $\rho \in \C^{d_0} \otimes \C^{d_1}$, we can run the weak-membership algorithm for $\Sep(d,d)$ on 
\[
(V_0\otimes V_1)\rho(V_0\otimes V_1)^*.
\]
Completeness is immediate, since local isometries map separable states to separable states. For soundness, it suffices to prove
\[
\min_{\psi\in\Sep(d,d)} \norm{(V_0\otimes V_1)\rho(V_0\otimes V_1)^*-\psi}_{\Frob} \geq 
\min_{\psi\in\Sep(d_0,d_1)} \norm{\rho-\psi}_{\Frob}.
\]
Indeed, fix any $\sigma\in\Sep(d,d)$. Its compression $(V_0\otimes V_1)^*\sigma(V_0\otimes V_1) \in \C^{d_0}\otimes\C^{d_1}$ is a subnormalized separable positive operator. Therefore,
\[(V_0\otimes V_1)^*\sigma(V_0\otimes V_1)
+
\left(
1-\Tr\left((V_0\otimes V_1)^*\sigma(V_0\otimes V_1)\right)
\right)
\frac{\Id_{d_0d_1}}{d_0d_1}
\]
is a state in $\Sep(d_0,d_1)$. Consequently,
\begin{align*}
\min_{\psi\in\Sep(d_0,d_1)} \norm{\rho-\psi}_{\Frob}^2
&\leq
\left\|
\rho-(V_0\otimes V_1)^*\sigma(V_0\otimes V_1)-
\left(
1-\Tr\left((V_0\otimes V_1)^*\sigma(V_0\otimes V_1)\right)
\right)
\frac{\Id_{d_0d_1}}{d_0d_1}
\right\|_{\Frob}^2\\
&=
\left\|
\rho-(V_0\otimes V_1)^*\sigma(V_0\otimes V_1)
\right\|_{\Frob}^2
-
\frac{
\left(
1-\Tr\left((V_0\otimes V_1)^*\sigma(V_0\otimes V_1)\right)
\right)^2
}{d_0d_1}\\
&\leq
\left\|
\rho-(V_0\otimes V_1)^*\sigma(V_0\otimes V_1)
\right\|_{\Frob}^2\\
&=
\left\|
(V_0\otimes V_1)^*
\left(
(V_0\otimes V_1)\rho(V_0\otimes V_1)^*-\sigma
\right)
(V_0\otimes V_1)
\right\|_{\Frob}^2\\
&\leq
\left\|
(V_0\otimes V_1)\rho(V_0\otimes V_1)^*-\sigma
\right\|_{\Frob}^2.
\end{align*}
The last inequality uses that compression by isometries is contractive in Frobenius norm. Taking the infimum over $\sigma\in\Sep(d,d)$ proves the desired inequality.
\section{Haar Moment Bounds}\label{sec:haar-bound}

For $a,b\in[d]$, define the matrix
\[
 H_{ab}:=u_bu_a^*-\delta_{ab}\frac{\Id}{d}.
\]
We prove our main analytic estimate in the following.
\begin{lemma}\label{lem:one-party}
Let \(\mathcal H\) be a finite-dimensional Hilbert space, let $\mathcal A:L(\C^d)\to\mathcal H$ be linear, where \(L(\C^d)\) is equipped with the Frobenius inner product and let $U \sim{\rm Haar}(\mathcal U(d))$. There exist absolute constants $c_0,c_1>0$ such that for all $a,b\in[d]$ and all $2\le q\le c_0d$
\[
 \left\|
 \mathcal A(UH_{ab}U^*)
 \right\|_{L_q(\mathcal H)}
 \le
 c_1\frac qd\norm{\mathcal A}_{\Frob}.
\]
\end{lemma}

\begin{proof}
First consider the rank-one traceless case $ H:=ww^*-\Id / d$,
where \(w\in\C^d\) is a fixed unit vector.  Then \(Uw\) is a
Haar-uniform unit vector, which we can alternatively sample as
\[
 Uw\stackrel{d}{=}  \frac{G}{\norm{G}_2},
 \quad\text{with}\quad
 G\sim\mathcal N_{\C}(0,\Id_d).
\]
Then
\begin{align*}
\norm{\mathcal A(UHU^*)}_{L_q(\mathcal H)} &= \norm{\norm{G}_2^{-2}\mathcal A\left(GG^*-\frac{\norm{G}_2^2}{d}\Id\right)}_{L_q(\mathcal H)} \\&\le
 \norm{\norm{G}_2^{-2}}_{L_{2q}}\cdot 
 \norm{\mathcal A \left(GG^*-\frac{\norm{G}_2^2}{d}\Id\right)}_{L_{2q}(\mathcal H)}
\end{align*}
by Hölder inequality. By \Cref{lem:inverse-gaussian-norm} (Gaussian inverse moment bound), for \(q\le d/4\),
\[
 \norm{\norm{G}_2^{-2}}_{L_{2q}}
 \le
 \frac2d.
\]
On the other hand, the factor on the RHS is the norm of a degree-2 polynomial map in $G$, thus by \Cref{lem:gaussian-hypercontractivity} (Gaussian hypercontractivity), we can bound it by
\[
 \norm{\mathcal A \left(GG^*-\frac{\norm{G}_2^2}{d}\Id\right)}_{L_{2q}(\mathcal H)}
 \le
 c_2 q\norm{\mathcal A \left(GG^*-\frac{\norm{G}_2^2}{d}\Id\right)}_{L_2(\mathcal H)}
\]
for some absolute constant $c_2$.
It remains to bound the \(L_2\)-norm.  Let \((h_\ell)_\ell\) be an orthonormal
basis of \(\mathcal H\) and let
\(\mathcal A^*:\mathcal H\to L(\C^d)\) denote the adjoint of \(\mathcal A\),
where \(L(\C^d)\) is equipped with the Frobenius inner product.  Then
\[
\begin{aligned}
\norm{\mathcal A \left(GG^*-\frac{\norm{G}_2^2}{d}\Id\right)}_{L_2(\mathcal H)}^2
&=
\sum_\ell
\E\left|
\Tr\left(
(\mathcal A^*h_\ell)^*
\left(
GG^*-\frac{\norm{G}_2^2}{d}\Id
\right)
\right)
\right|^2  \\
&=
\sum_\ell
\E\left|
G^*
\left(
(\mathcal A^*h_\ell)^*
-\frac{\Tr((\mathcal A^*h_\ell)^*)}{d}\Id
\right)
G
\right|^2  \\
&=
\sum_\ell
\left\|
(\mathcal A^*h_\ell)^*
-\frac{\Tr((\mathcal A^*h_\ell)^*)}{d}\Id
\right\|_{\Frob}^2 \\
&\le
\sum_\ell
\|\mathcal A^*h_\ell\|_{\Frob}^2
=
\|\mathcal A\|_{\Frob}^2 .
\end{aligned}
\]
The first equality uses Parseval’s identity in $\mathcal H$ and the definition
of the adjoint $\mathcal A^*$. The second equality uses
$\Tr(TGG^*)=G^*TG$ and $\|G\|_2^2=G^*G$. The third equality uses
\Cref{lem:complex-gaussian-quadratic-form} (Wick's formula), applied to the traceless operator
inside the quadratic form. The inequality uses that subtracting
\(\frac{\Tr T}{d}\Id\) is the Hilbert–Schmidt orthogonal projection onto the
traceless subspace, and hence cannot increase the Frobenius norm.
Combining these estimates proves the claim when $a=b$.

Now suppose \(a\ne b\), for $\omega\in\{1,-1,i,-i\}$, let
\[
 \psi_\omega:=\frac{u_a+\omega \cdot u_b}{\sqrt2}.
\]
Then we have the standard identity
\[
 u_bu_a^*
 =
 \frac12\left(\psi_1\psi_1^* 
 -\psi_{-1}\psi_{-1}^*\right)
 +
 \frac{1}{2i}\left(\psi_i\psi_i^*
 -\psi_{-i}\psi_{-i}^*\right).
\]
Since the coefficients sum to zero, we may subtract \(\Id/d\),
\[
 u_bu_a^*
 =
 \frac12\left(\psi_1\psi_1^*-\frac{\Id}{d}\right)
 -\frac12\left(\psi_{-1}\psi_{-1}^*-\frac{\Id}{d}\right)
 +\frac{1}{2i}\left(\psi_i\psi_i^*-\frac{\Id}{d}\right)
 -\frac{1}{2i}\left(\psi_{-i}\psi_{-i}^*-\frac{\Id}{d}\right).
\]
After conjugating by \(U\), each term is of the form discussed above. The same bound (up to a factor $2$) follows by the triangle inequality.
\end{proof}
Next, we apply the above estimate to obtain a bound on the norm of the entries of a matrix with vanishing partial traces.

\begin{proposition}\label{prop:haar-moment}
Let \(M\in L(\C^d\otimes\C^d)\) satisfy $\Tr_0 M=\Tr_1 M=0$ and let \(U,V\sim{\rm Haar}(\mathcal U(d))\) independently.  There are absolute
constants \(c_0,c_1>0\) such that, for every \(a,b,c,e\in[d]\) and every integer
\(p\) with \(1\le p\le c_0d\),
\[
\norm{(u_a\otimes u_c)^*(U\otimes V)^*M(U\otimes V)(u_b\otimes u_e)}_{L_{2p}}
 \le
 c_1\cdot \frac{p^2}{d^2}\norm{M}_{\Frob}.
\]
\end{proposition}
\begin{proof}
Let \(\hat c_0,\hat c_1>0\) be the constants from \cref{lem:one-party}.  Choose the constant
\(c_0\) in the proposition so that \(2c_0\le \hat c_0\) and set $ q:=2p$.
Then \(p\le c_0d\) implies \(q\le \hat c_0d\), so \cref{lem:one-party} may be
applied with exponent \(q\). Expand the expression
\begin{align*}
(u_a\otimes u_c)^*(U\otimes V)^*M(U\otimes V)(u_b\otimes u_e) 
&=\Tr\!\left(
 M\,
 \left(Uu_bu_a^*U^*\right)
 \otimes
 \left(Vu_eu_c^*V^*\right)
 \right)\\
 &= \Tr\!\left(
 M(UH_{ab}U^*)\otimes(VH_{ce}V^*)
 \right) 
\\&\quad+
 \frac{\delta_{ec}}{d}
 \Tr\!\left(
 M(UH_{ab}U^*)\otimes\Id
 \right)
\\&\quad+
 \frac{\delta_{ab}}{d}
 \Tr\!\left(
 M\Id\otimes(VH_{ce}V^*)
 \right)
+
 \frac{\delta_{ab}\delta_{ec}}{d^2}
 \Tr M
 \\
 &=\Tr\!\left(
 M(UH_{ab}U^*)\otimes(VH_{ce}V^*)
 \right)
\end{align*}
where in the last line we used that the last three summands vanish because $\Tr_0 M = \Tr_1 M =0$.
Now fix $V$, and define the operator $\mathcal{L}_V :L(\C^d)\to \C$ as
\[
 \mathcal{L}_V(S)
 :=
 \Tr\!\left(
 M(S\otimes VH_{ce}V^*)
 \right)
 \qquad \forall S\in L(\C^d),
\]
 a scalar linear map. Now, by \Cref{lem:one-party},
 \[
 \norm{\mathcal{L}_V(UH_{ab}U^*)}_{L_q} \leq \hat{c}_1 \frac{q}{d}\norm{\mathcal{L}_V}_{\Frob}.
 \]
Taking $q$-powers, expectation over $V$, and again $1/q$-powers, we obtain 
\[
\norm{\Tr\!\left(
 M(UH_{ab}U^*)\otimes(VH_{ce}V^*)
 \right)}_{L_q} \leq \hat c_1 \frac{q}{d} \left(
\E_V
\|\mathcal L_V\|_{\Frob}^q
\right)^{1/q}.
\] 
It remains to bound the term on the RHS.  Define the linear map $ \mathcal L:L(\C^d)\to L(\C^d)$ by its coefficients
\[
 \mathcal L(T)_{\alpha\beta}
 :=
 \Tr\!\left(M(u_\alpha u_\beta^*\otimes T)\right) \qquad \forall T\in L(\C^d).
\]
Then, by the definition of the Frobenius norm of the scalar map
\(\mathcal L_V\),
\[
\begin{aligned}
 \|\mathcal L_V\|_{\Frob}^2
 =
 \sum_{\alpha,\beta}
 \left|\mathcal L_V(u_\alpha u_\beta^*)\right|^2 =
 \sum_{\alpha,\beta}
 \left|
 \Tr\!\left(
 M(u_\alpha u_\beta^*\otimes VH_{ce}V^*)
 \right)
 \right|^2 =
 \|\mathcal L(VH_{ce}V^*)\|_{\Frob}^2 .
\end{aligned}
\]
Therefore,
\[
\left(
\E_V
\|\mathcal L_V\|_{\Frob}^q
\right)^{1/q}=
\left(
\E_V
\|\mathcal L(VH_{ce}V^*)\|_{\Frob}^q
\right)^{1/q}
=
\|\mathcal L(VH_{ce}V^*)\|_{L_q(L(\C^d))}.
\]
Now apply \Cref{lem:one-party} with Hilbert space
\(\mathcal H=L(\C^d)\) and linear map \(\mathcal L\). This gives
\[
\|\mathcal L(VH_{ce}V^*)\|_{L_q(L(\C^d))}
\le
\hat c_1\frac qd \|\mathcal L\|_{\Frob} =
\hat c_1\frac qd \|M\|_{\Frob} 
\]
 since \(\mathcal L\) is just
a reshuffling of the coefficients of \(M\).
\end{proof}

\subsection{Flat Decomposition of Hermitian Matrices}\label{sec:flat}
We conclude our Haar-moment analysis by proving the existence of a Hermitian decomposition into a \emph{flat} component plus a component with small Frobenius norm. This holds for matrices with vanishing partial traces.

\begin{lemma}\label{lem:trunc-spread}
Let \(M\in L(\C^d\otimes\C^d)\) be Hermitian and satisfy $ \Tr_0 M=\Tr_1 M=0$
and let \(U,V\sim{\rm Haar}(\mathcal U(d))\) independently. For every \(0<\delta<1/2\), with probability at least \(1-\delta^2\), there
is an efficiently-computable Hermitian decomposition
\[
 (U\otimes V)^*M(U\otimes V)=M_{\rm flat}+M_{\rm tail}
\]
such that:
\[
\norm{M_{\rm tail}}_{\Frob}\le \delta\norm{M}_{\Frob} \quad\text{and}\quad\norm{M_{\rm flat}}_{{\rm max}}
 \le
 c_{\rm flat}
 \frac{(1+\log(1/\delta))^2}{d^2}\norm{M}_{\Frob}
\]
for an absolute constant $c_{\rm flat} >0$.
\end{lemma}
\begin{proof}
Let $c_0,c_1$ be the constants in \cref{prop:haar-moment}. If \(M=0\), the claim is trivial, so we assume without loss of generality that  \(M\neq 0\). Choose once and for all an integer \(m\ge2\) such that
\begin{equation}\label{eq:m-choice-trunc}
 64 c_1^2(m+1)^4\,8^{-2m}\le \frac12 .
\end{equation}
If $d < \frac{(m+1)(\log (1/\delta) +1)}{c_0}$, then we simply take
\[
M_{\rm flat} :=(U\otimes V)^*M(U\otimes V)\quad\text{and}\quad
M_{\rm tail} := 0.
\]
Then
\[
 \norm{M_{\rm flat}}_{{\rm max}}
 \le
 \norm{(U\otimes V)^*M(U\otimes V)}_{\Frob}
 =
 \norm{M}_{\Frob}
\]
by unitary invariance of the Frobenius norm, and the claim holds with $c_{\rm flat} := \frac{({m+1})^2}{{c_0}^2}$, simply because 
\[
\left(\frac{m+1}{c_0}\right)^2\frac{(\log (1/\delta) +1)^2}{d^2} >1.
\]
Suppose instead that $d \geq \frac{(m+1)(\log (1/\delta) +1)}{c_0}$, let $p$ be an integer with 
\begin{equation}\label{eq:p}
   m(1+\log(1/\delta))\leq p\leq (m+1)(1+\log(1/\delta)). 
\end{equation}
Notice that this implies that $p\le c_0d$ and therefore \Cref{prop:haar-moment} applies with this value of \(p\). Fix some $a,b,c,e\in[d]$ and let $Z$ denote the random variable for the corresponding entry of $(U\otimes V)^*M(U\otimes V)$. By Markov inequality and \Cref{prop:haar-moment}, we obtain
\[
\E \abs{Z}^2 \cdot \mathbbm{1}_{\{\abs{Z} >\tau\} } \leq \frac{\E \abs{Z}^{2p}}{\tau^{2p-2}} \leq \frac{(c_1p^2\norm{M}_{\Frob})^{2p}}{d^{4p}\tau^{2p-2}}.
\]
Taking $\tau := 8c_1p^2\norm{M}_{\Frob}/d^2$, we expand
\begin{align}
\E \abs{Z}^2 \cdot \mathbbm{1}_{\{\abs{Z} >\tau\} } &\leq \left(\frac{c_1p^2\norm{M}_{\Frob}}{d^2}\right)^2\cdot8^{-2p+2}\notag \\&\leq \frac{64 \norm{M}_{\Frob}^2c_1^2(m+1)^4(1+\log(1/\delta))^48^{-2m(1+\log(1/\delta))}}{d^4}\notag\\
&\le \frac{1}{2}\frac{\norm{M}_{\Frob}^2}{d^4}\delta^4\label{eq:single-entry-tail}
\end{align}
where in the second inequality we used \Cref{eq:p} and
in the last inequality we used that $m\geq 2$ to bound
\[
(1+\log(1/\delta))^4\,8^{-2m(1+\log(1/\delta))}
\le
8^{-2m}e^{-4\log(1/\delta)}
=
8^{-2m}\delta^4
\]
and then appealed to \Cref{eq:m-choice-trunc}.
Define \(M_{\rm tail}\) entrywise by
\[
 (M_{\rm tail})_{ac,be}
 :=
 (U\otimes V)^*M(U\otimes V)_{ac,be}
 \cdot \mathbbm{1}_{\{|(U\otimes V)^*M(U\otimes V)_{ac,be}|>\tau\}},
\]
and set $M_{\rm flat}:=(U\otimes V)^*M(U\otimes V)-M_{\rm tail}$. Because the matrix $(U\otimes V)^*M(U\otimes V)$ is Hermitian, if the entry indexed by $(ac,be)$ passes the threshold, then so does the entry $(be,ac)$. It follows that $M_{\rm tail}$ is Hermitian, and so is $M_{\rm flat}$. By construction and using \Cref{eq:p},
\[
\norm{M_{\rm flat}}_{\rm max} \leq \tau =8c_1p^2\norm{M}_{\Frob}/d^2 \leq 8c_1(m+1)^2(1+\log(1/\delta))^2\norm{M}_{\Frob}/d^2
\]
so the claim holds with $c_{\rm flat}:= 8c_1(m+1)^2$. It remains to prove the Frobenius tail bound with high probability. By \Cref{eq:single-entry-tail}, in expectation,
\[
\E \norm{M_{\rm tail}}_{\Frob}^2 \leq \frac{1}{2}\norm{M}_{\Frob}^2\delta^4
\]
and the high-probability claim follows by a Markov argument. 
\end{proof}

\section{Discretization}\label{sec:disc-main}

Let $M\in L(\C^d\ot \C^d)$ be a Hermitian matrix with $\Tr_0 M = \Tr_1 M =0$. Throughout this section we assume that $d^2 \norm{M}_{\rm max} \leq \gamma$.
Define the quartic form
\begin{equation}\label{eq:quartic}
Q_M(x,y):=\frac{1}{d^2} (x\ot y)^*M(x\ot y)=\frac1{d^2}\sum_{a,b,c,e}M_{ac,be}{x}^*_ax_b{y}^*_cy_e \in \mathbb{R}.
\end{equation}
Because $M$ is Hermitian, this quartic form is always real-valued.

\subsection{Flat Matrices Have Flat Near-Optimizers}\label{sec:bounded-amp}

We prove that under the above conditions, there exist near-optimizers with small infinity norm for the quartic form in \Cref{eq:quartic}.

\begin{lemma}
 \label{lem:bounded-amplitude}
For every \(0<\eta<1\), there exist polynomials $\kappa, \Delta=\poly(\gamma ,1/\eta)$ such that, if $d \geq \Delta$, then 
\[
\max_{\norm{x}_2=\norm{y}_2=\sqrt d}Q_M(x,y) \leq 
\max_{\substack{\norm{x}_2=\norm{y}_2=\sqrt d\\
\norm{x}_\infty,\norm{y}_\infty\le \kappa}} Q_M(x,y) +\eta.
\]
\end{lemma}

\begin{proof}
If $\gamma=0$, then $M=0$ and the conclusion is immediate. We henceforth assume $\gamma>0$. We prove the claim by showing that, for every feasible pair $x,y$, there is another feasible pair $\hat{x}, \hat{y}$ with $\norm{\hat x}_\infty,\norm{\hat y}_\infty \leq \kappa$ such that
    \[
    Q_M(x,y) \leq Q_M(\hat x,\hat y) + \eta.
    \]
Let $\sigma$ be a threshold value to be determined later, define the truncated vectors $\tilde x,\tilde y$ coordinatewise
    \[
    \tilde x_a := x_a \mathbbm{1}_{\{\abs{x_a} \leq \sigma\}} \quad\text{and}\quad
    \tilde y_c := y_c \mathbbm{1}_{\{\abs{y_c} \leq \sigma\}}.
    \]
Using \(d^2\norm{M}_{\rm max}\le \gamma\), we bound
\[
\begin{aligned}
&\abs{Q_M(x,y)-Q_M(\tilde x,\tilde y)} \\
&=
\left|
\frac1{d^2}
\sum_{a,b,c,e}
M_{ac,be}
\left(
{x}^*_ax_b{y}^*_cy_e
-
\tilde{x}^*_a\tilde x_b\tilde{y}^*_c\tilde y_e
\right)
\right| 
\\
&\le
\left|
\frac1{d^2}
\sum_{a,b,c,e}
M_{ac,be}
\left(
{x}^*_ax_b-\tilde{x}_a^*\tilde x_b
\right)
{y}^*_cy_e
\right| +
\left|
\frac1{d^2}
\sum_{a,b,c,e}
M_{ac,be}
\tilde{x}^*_a\tilde x_b
\left(
{y}^*_cy_e-\tilde{y}^*_c\tilde y_e
\right)
\right| \\
&\le
\frac{\gamma}{d^4}
\left(
 \norm{x-\tilde x}_1\norm{x}_1+\norm{\tilde x}_1\norm{x-\tilde x}_1
\right)
\norm{y}_1^2 +
\frac{\gamma}{d^4}
\norm{\tilde x}_1^2
\left(
 \norm{y-\tilde y}_1\norm{y}_1+\norm{\tilde y}_1\norm{y-\tilde y}_1
\right) \\
&\le
\frac{\gamma}{d^4}
\left(
 \frac d\sigma\cdot d+d\cdot\frac d\sigma
\right)d^2
+
\frac{\gamma}{d^4}
d^2
\left(
 \frac d\sigma\cdot d+d\cdot\frac d\sigma
\right) =
\frac{4\gamma}{\sigma}
\end{aligned}
\]
where the last line uses $\norm{x}_1,\norm{\tilde x}_1 \leq \sqrt{d}\norm{x}_2 \leq d$ and, by Cauchy--Schwarz,
\[
\norm{x - \tilde{x}}_1\le \sqrt{\abs{{\rm Supp}(x - \tilde{x})}}\norm{x - \tilde{x}}_2 \le \sqrt{\frac{d}{\sigma^2}}\sqrt{d} =\frac{d}{\sigma}.
\]
Next, we refill the missing norm of the truncated vectors. Let \(\xi,\zeta\in\C^d\) be independent random vectors where each coordinate is sampled iid from the unit circle. Define
\[
\breve x:= \tilde x + \sqrt{1-\frac{\norm{\tilde x}_2^2}{d}}\xi
\quad\text{and}\quad
\breve y:= \tilde y + \sqrt{1-\frac{\norm{\tilde y}_2^2}{d}}\zeta
\]
then clearly $\norm{\breve x}_\infty,\norm{\breve y}_\infty\leq \sigma +1$. 
Independence of \(\xi\) and \(\zeta\) gives
\begin{equation}\label{eq:cancellations2}
 \E\!\left(
 \breve{x}^*_a \breve{x}_b
 \breve{y}^*_c \breve{y}_e
 \right)
 =
 \E\!\left(
 \breve{x}^*_a \breve{x}_b\right)\E\!\left(
 \breve{y}^*_c \breve{y}_e
 \right)
 =
 \left(\tilde{x}^*_a\tilde x_b+\left(1 -\frac{\norm{\tilde x}_2^2}{d} \right)\delta_{ab}\right)
 \left(\tilde{y}^*_c \tilde y_e+\left(1 -\frac{\norm{\tilde y}_2^2}{d} \right)\delta_{ce}\right)
\end{equation}
since $\E {\xi}^*_a\xi_b = \delta_{ab}$ and $\E \xi_a = 0$, and similarly for $\zeta$.
We claim that such refilling does not change the value of the quartic form, in expectation. Indeed,
 by \Cref{eq:cancellations2},
\[
\begin{aligned}
 \E Q_M(\breve x,\breve y)
 &=
 Q_M(\tilde x, \tilde y) +
 \frac1{d^2}\left(1-\frac{\norm{\tilde x}_2^2}{d}\right)
 \sum_{a,c,e}
 M_{ac,ae}\,
 \tilde{y}^*_c\tilde y_e\\
 &\quad+
 \frac1{d^2}\left(1-\frac{\norm{\tilde y}_2^2}{d}\right)
 \sum_{a,b,c}
 M_{ac,bc}\,
 \tilde{x}^*_a\tilde x_b+
 \frac1{d^2}\left(1-\frac{\norm{\tilde x}_2^2}{d}\right)
 \left(1-\frac{\norm{\tilde y}_2^2}{d}\right)
 \sum_{a,c}M_{ac,ac}\\
 &=
 Q_M(\tilde x,\tilde y).
\end{aligned}
\]
where the second summand vanishes because $\sum_a M_{ac,ae}=0$ for all $c,e$ (as a consequence of $\Tr_0 M=0$), the third summand vanishes because $\sum_c M_{ac,bc}=0$ for all $a,b$ (as a consequence of $\Tr_1 M=0$), and the fourth one vanishes by either identity.

Next, we argue that, with good probability, the refilled vectors have nearly correct norms. Expanding,
\[
 \norm{\breve x}_2^2
 =
 \norm{\tilde x}_2^2+\left(1-  \frac{\norm{\tilde x}_2^2}{d}\right) d
 +2\sqrt{1-  \frac{\norm{\tilde x}_2^2}{d}}\,\operatorname{Re}\ip{\tilde x}{\xi}
 =
 d+2\sqrt{1-\frac{\norm{\tilde x}_2^2}{d}}\,\operatorname{Re}\ip{\tilde x}{\xi}.
\]
Then Chebyshev's inequality gives, for every \(t>0\),
\begin{equation}
\label{eq:norm-good-u}
 \Pr\left(
 \left|\frac{\norm{\breve x}_2^2}{d}-1\right|>t
 \right)=
 \Pr\left(
 \left|\frac{2}{d}\sqrt{1-\frac{\norm{\tilde x}_2^2}{d}}\,\operatorname{Re}\ip{\tilde x}{\xi}\right|>t
 \right)
 \le
 \frac{4}{t^2d}
\end{equation}
where we used that $\E\abs{\ip{\tilde x}{\xi}}^2=\norm{\tilde x}_2^2\le d$ and therefore 
\[\mathbb E\left|
\frac{2}{d}\sqrt{1-\frac{\|\tilde x\|_2^2}{d}}\operatorname{Re}\langle \tilde x,\xi\rangle
\right|^2
\le
\frac4d.\]
The same estimate works for $\breve y$, so the probability that both deviations are at most $t$ is at least $1 - \frac{8}{t^2d}$. 
Call such a good event $E$. We then argue that, conditioned on $E$, the refilling does not change the value of the quartic form by a lot. 
First, notice that 
\begin{equation}\label{eq:bound}    
\abs{Q_M(\tilde x,\tilde y)},\abs{Q_M(\breve x,\breve y)}\leq \gamma (\sigma+1)^4
\end{equation}
Indeed, for any \(v,w\) with
\(\norm{v}_\infty,\norm{w}_\infty\le \sigma+1\), we have
$\abs{Q_M(v,w)}
\le
\frac{\gamma}{d^4}\norm{v}_1^2\norm{w}_1^2
\le
\gamma(\sigma+1)^4$.
Therefore
\begin{align*}
\E\left(
Q_M(\breve{x},\breve{y})
\mid E
\right) 
&=
\frac{
\E\left(
\left(Q_M(\breve{x},\breve{y})\right)\mathbbm 1_E\right)
}{\Pr(E)} =
\frac{
Q_M(\tilde{x},\tilde{y})
-
\E\left(
\left(Q_M(\breve{x},\breve{y})\right)\mathbbm 1_{\bar{E}}\right)
}{\Pr(E)} \\
&\ge
\frac{
Q_M(\tilde{x},\tilde{y})
-
\gamma(\sigma+1)^4\Pr(\bar{E})
}{\Pr(E)}\ge Q_M(\tilde{x},\tilde{y})
-
\frac{
16\gamma(\sigma+1)^4
}{
t^2d\left(1-\frac{8}{t^2d}\right)
}
\end{align*}
using the fact that, in expectation, the refilling does not change the value of the quartic form and \Cref{eq:bound} twice.
Hence, if $\frac{8}{t^2d}\le \frac12$, there must exist an outcome in $E$ such that 
\[
Q_M(\breve{x},\breve{y})
\ge
Q_M(\tilde{x},\tilde{y})
-
\frac{
32\gamma(\sigma+1)^4
}{
t^2d
}.
\]
Fix such $(\breve{x},\breve{y})$ and define
\[
\hat x := \frac{\sqrt{d}}{\norm{\breve{x}}_2} \breve{x} \quad\text{and}\quad
\hat y := \frac{\sqrt{d}}{\norm{\breve{y}}_2} \breve{y}.
\]
By definition of $E$, we know that $1-t
\le
\frac{\norm{\breve{x}}_2^2}{d},
\frac{\norm{\breve{y}}_2^2}{d}
\le
1+t$, consequently
\[
\|\hat x\|_\infty
=
\frac{\sqrt d}{\|\breve{x}\|_2}\|\breve{x}\|_\infty
\le
\frac{\sqrt d}{\sqrt{d(1-t)}}\|\breve{x}\|_\infty
=
\frac{\|\breve{x}\|_\infty}{\sqrt{1-t}} \leq \frac{\sigma+1}{\sqrt{1-t}}
\]
and similarly for $\norm{\hat y}_\infty$.
Assuming that \(t\le1/2\), \Cref{eq:bound} gives
\[
\begin{aligned}
Q_M(\hat x,\hat y)
&=
\frac{d^2}{\norm{\breve{x}}_2^2\norm{\breve{y}}_2^2}Q_M(\breve{x},\breve{y})\ge
Q_M(\breve{x},\breve{y})
-
\left|\frac{d^2}{\norm{\breve{x}}_2^2\norm{\breve{y}}_2^2}-1\right|\abs{Q_M(\breve{x},\breve{y})}\ge
Q_M(\breve{x},\breve{y})
-
8t\,\gamma(\sigma+1)^4
\end{aligned}
\]
where we used
\[
\left|\frac{1}{xy}-1\right|\le 8t
\quad\text{for }x,y\in[1-t,1+t],\ t\le1/2.
\]
Combining with the above estimates, we obtain that
\begin{align*}
Q_M(\hat x,\hat y)\geq Q_M( x, y) 
-
\frac{4\gamma}{\sigma}
-\left(\frac{32}{t^2d} + 8t\right)\gamma(\sigma+1)^4.
\end{align*}
It remains to set the parameters. We choose:
\begin{itemize}
    \item  $\sigma:=\max\left\{1,\frac{12\gamma}{\eta}\right\}$,
    \item $t:= \min\left\{
\frac12,\,
\frac{\eta}{
24\gamma(\sigma+1)^4
}
\right\}$, and
\item $d\ge \Delta:=
\max\left\{
\frac{16}{t^2},\,
\frac{
96\gamma(\sigma+1)^4
}{
\eta t^2
}
\right\}$.
\end{itemize}
With these parameter choices, $Q_M(\hat x,\hat y)
\ge
Q_M(x,y)-\eta$.
Moreover, $\norm{\hat x}_2=\norm{\hat y}_2=\sqrt d$ by construction, and 
\[
\norm{\hat x}_\infty,\norm{\hat y}_\infty
\le
\frac{\sigma+1}{\sqrt{1-t}}
\le
\sqrt{2}(\sigma+1).
\]
Thus the claim holds with $\kappa:=\sqrt{2}(\sigma+1) =\poly(\gamma,1/\eta)$ and $\Delta = \poly(\gamma,1/\eta)$.
\end{proof}

\subsection{Dimension-Free Discretization}\label{sec:discr}

Let $\mathbb{D} \subset \C$ be the (closed) unit disk in the Euclidean metric. For a parameter $\varepsilon >0 $, let $\mathcal{E}$ be an $\varepsilon$-net of $\kappa\mathbb{D}$ contained in $\kappa\mathbb{D}$. For instance, choosing the (suitably scaled) square lattice and adding an $\varepsilon$-net on the boundary, we can choose $\mathcal E$ such that 
\[
\abs{\mathcal{E}} \leq c_{\rm net} \max\left\{1,\frac{\kappa^2}{\varepsilon^2}\right\},
\]
for an absolute constant $c_{\rm net}$. Let $w,x,y,z\in \C^d$ satisfy the following properties:
\begin{itemize}
\item $\norm{w}_\infty, \norm{x}_\infty, \norm{y}_\infty, \norm{z}_\infty \leq \kappa$,
\item $\norm{x-w}_\infty \leq \varepsilon$, and
\item $\norm{y-z}_\infty \leq \varepsilon$.
\end{itemize}
We show that the quartic form in \Cref{eq:quartic} is Lipschitz-continuous in their inputs.
\begin{lemma}\label{lem:obj-lipschitz}
    $\abs{Q_M(x,y) - Q_M(w,z)} \leq 4 \gamma \kappa^3 \varepsilon$.
\end{lemma}
\begin{proof}
By the triangle inequality
\[
\begin{aligned}
&\left|
x_a^*x_by_c^*y_e-
w_a^*w_bz_c^*z_e
\right|\\
&=
\left|
(x_a^*x_b-w_a^*w_b)y_c^*y_e
+
w_a^*w_b(y_c^*y_e-z_c^*z_e)
\right|\\
&\le
\left(
|x_a-w_a|\,|x_b|
+
|w_a|\,|x_b-w_b|
\right)
|y_c|\,|y_e|
+
|w_a|\,|w_b|
\left(
|y_c-z_c|\,|y_e|
+
|z_c|\,|y_e-z_e|
\right)\\
&\le
(\varepsilon \kappa+\varepsilon \kappa)\kappa^2
+
\kappa^2(\varepsilon \kappa+\varepsilon \kappa)
=
4\kappa^3\varepsilon
\end{aligned}
\]
and therefore $\abs{Q_M(x,y)-Q_M(w,z)} \leq 4\gamma \kappa^3\varepsilon$.
\end{proof}
For $\lambda >0$, define the \emph{penalized} objective function as
\begin{equation}\label{eq:penalized}
\Lambda_{M}(x,y) := Q_M(x,y) -\lambda\left( \frac1d \sum_{a}\abs{x_a}^2-1\right)^2-\lambda\left( \frac1d \sum_{c}\abs{y_c}^2-1\right)^2.
\end{equation}
We prove a similar continuity bound of the penalized objective function.
\begin{lemma}\label{lem:lambda-Lipschitz}
$\abs{\Lambda_{M}(x,y) - \Lambda_{M}(w, z)}\leq 4 \gamma \kappa^3 \varepsilon + 8\lambda \kappa(\kappa^2+1)\varepsilon$.
\end{lemma}
\begin{proof}
We bound the difference between the penalty terms. First notice that,
\[
 \abs{\frac1d\sum_a |x_a|^2-\frac1d\sum_a |w_a|^2}
 \le
  \frac1d\sum_a\abs{ |x_a|^2-|w_a|^2}
  \le
\frac1d\sum_a
    |x_a-w_a|(|x_a|+|w_a|)
  \le
 2\kappa\varepsilon.
\]
Consequently, expanding the difference between the first penalty terms we obtain
\[
\begin{aligned}
&\abs{\lambda\left(\frac1d\sum_a |x_a|^2-1\right)^2-\lambda\left(\frac1d\sum_a |w_a|^2-1\right)^2}\\ &\leq
 \lambda\abs{\frac1d\sum_a |x_a|^2-\frac1d\sum_a |w_a|^2}
 \left(
 \abs{\frac1d\sum_a |x_a|^2-1}+\abs{\frac1d\sum_a |w_a|^2-1}
 \right) \\
&\le
 2\lambda \kappa\varepsilon\cdot 2(\kappa^2+1)
 =
 4\lambda \kappa(\kappa^2+1)\varepsilon.
\end{aligned}
\]
The same calculation gives the exact same bound for the other penalty term, and the difference between the quartic forms is bounded by \Cref{lem:obj-lipschitz}.
\end{proof}
We conclude this section by showing that restricting solutions to have coefficients in the $\varepsilon$-net, does not change the value of the objective function substantially.

\begin{lemma}\label{lem:dimension-free-discretization}
For every \(0<\eta<1\), there are parameters $\kappa,\lambda,1/\varepsilon, \Delta
=
\poly(\gamma,1/\eta)$ such that, for all $d\geq \Delta$,
\[
\left|
\max_{\|x\|_2=\|y\|_2=\sqrt d}Q_M(x,y)
-
\max_{x,y\in\mathcal E^d}\Lambda_{M}(x,y)
\right|
\le \eta .
\]
\end{lemma}
\begin{proof}
Apply \Cref{lem:bounded-amplitude} with accuracy \(\eta/4\), obtaining
 \(\kappa,\Delta=\poly(\gamma,1/\eta)\) such that, for all \(d\geq \Delta\),
 \begin{equation}\label{eq:recall}     
 \max_{\|x\|_2=\|y\|_2=\sqrt d} Q_M(x,y) \leq 
 \max_{\substack{\|x\|_2=\|y\|_2=\sqrt d\\
\|x\|_\infty,\|y\|_\infty\le \kappa}} Q_M(x,y)
+\eta/4.
 \end{equation}
Choose 
\begin{itemize}
    \item $\lambda
:=
1+\frac{\gamma^2+\gamma^2 \kappa^4}{\eta}$ and
\item $\varepsilon := \frac{\eta}{16\gamma\kappa^3
+
32\lambda \kappa(\kappa^2+1)}$.
\end{itemize}
First we prove the lower bound. Let \(x,y\) be a maximizer of the RHS in \Cref{eq:recall}. 
Choose \(w,z\in\mathcal E^d\) by rounding them coordinatewise, so that
\[
\|x-w\|_\infty,\|y-z\|_\infty\le \varepsilon.
\]
Since \(\norm{x}_\infty,\norm{y}_\infty,\norm{w}_\infty,\norm{z}_\infty\) are all bounded by \(\kappa\), \Cref{lem:lambda-Lipschitz} gives $|\Lambda_{M}(x,y)-\Lambda_{M}(w,z)| \le \eta/4$. But notice that \(\Lambda_{M}(x,y)= Q_M(x,y)\),  because \(x,y\) are exactly normalized. Hence the lower bound follows with $\eta/2$.

It remains to prove the upper bound. Fix arbitrary \(x,y\in\mathcal E^d\), and define
\[
 \tilde{x}:=
 \begin{cases}
 \sqrt d\,x/\norm{x}_2, & x\neq 0,\\
 \sqrt d\,u_1, & x=0,
 \end{cases}
 \quad\text{and}\quad
 \tilde{y}:=
 \begin{cases}
 \sqrt d\,y/\norm{y}_2, & y\neq 0,\\
 \sqrt d\,u_1, & y=0.
 \end{cases}
\]
\begin{claim}\label{claim:1}
    $Q_M(\tilde x, \tilde y) \geq \Lambda_{M}(x,y) -\frac{ \gamma^2 + \gamma^2\kappa^4}{4\lambda}$.
\end{claim}
Note that $\tilde x$ and $\tilde y$ are exactly normalized, so this implies that the same bound holds for the maximizers of the objective function, and the bound follows by our choice of $\lambda$.

It remains to prove \Cref{claim:1}. Indeed,
\[
\begin{aligned}
Q_M(x,y) &=  \frac{\|x\|_2^2}{d}\frac{\|y\|_2^2}{d}Q_M(\tilde x, \tilde y)\le Q_M(\tilde x,\tilde  y) +
\gamma\left|
\frac{\|x\|_2^2}{d}\frac{\|y\|_2^2}{d}-1
\right|\\
&\le Q_M(\tilde x,\tilde  y)
+
\gamma
\left|\frac{\|x\|_2^2}{d}-1\right|
+
\gamma \kappa^2
\left|\frac{\|y\|_2^2}{d}-1\right|.
\end{aligned}
\]
The first inequality follows from the bound $|Q_M(\tilde x,\tilde y)|\le \gamma$, which in turn follows using \(\|\tilde x\|_1,\|\tilde y\|_1\le d\) and the entrywise bounds on \(M\).
The second inequality uses 
\[
\left|
\frac{\|x\|_2^2}{d}\frac{\|y\|_2^2}{d}-1
\right|
=
\left|
\frac{\|x\|_2^2}{d}
\left(\frac{\|y\|_2^2}{d}-1\right)
+
\left(\frac{\|x\|_2^2}{d}-1\right)
\right|
\le
\left|\frac{\|x\|_2^2}{d}-1\right|
+
\kappa^2\left|\frac{\|y\|_2^2}{d}-1\right|
\]
which holds because $\frac{\|x\|_2^2}{d} \le \kappa^2$. Consequently,
\[
\begin{aligned}
\Lambda_{M}(x,y)
&\le
Q_M(\tilde x,\tilde  y) +
\gamma
\left|\frac{\|x\|_2^2}{d}-1\right|
-\lambda\left(\frac{\|x\|_2^2}{d}-1\right)^2 +
\gamma \kappa^2
\left|\frac{\|y\|_2^2}{d}-1\right|
-\lambda\left(\frac{\|y\|_2^2}{d}-1\right)^2\\
&\le
Q_M(\tilde x,\tilde  y)
+
\frac{\gamma^2+\gamma^2 \kappa^4}{4\lambda}
\end{aligned}
\]
where the second inequality uses \(ab-\lambda b^2\le a^2/(4\lambda)\) for any \(a,b\ge0\).
\end{proof}

\section{The Main Algorithm}\label{sec:algo}

\subsection{Dense Constraint Satisfaction Problems}\label{sec:csp}

We recall the definition of a dense Max-\(k\)-CSP instance.

\begin{definition}[Dense Max-\(k\)-CSP]
Let \(n,k\in\mathbb N\).  A dense Max-\(k\)-CSP instance $\mathcal I$ consists of:
\begin{itemize}
    \item Variables \(z_1,\ldots,z_n\), each taking values in a finite alphabet
    \(\Omega\).
    \item For each ordered \(k\)-tuple \((i_1,\ldots,i_k)\in[n]^k\), a predicate $P_{i_1,\ldots,i_k}:\Omega^k\to[0,1]$.
\end{itemize}
Repeated indices are allowed.  The value of an assignment
\(z=(z_1,\ldots,z_n)\in\Omega^n\) is
\[
 \val_{\mathcal I}(z)
 :=
 \frac1{n^k}
 \sum_{i_1,\ldots,i_k\in[n]}
 P_{i_1,\ldots,i_k}(z_{i_1},\ldots,z_{i_k}).
\]
The optimum value is $\OPT(\mathcal I):=\max_{z\in\Omega^n}\val_{\mathcal I}(z)$.
\end{definition}
The following shows that dense Max-\(k\)-CSP admit a polynomial-time approximation scheme (PTAS) when $k$ and $\abs{\Omega}$ are constant. This was shown first in 
\cite{Arora1995DensePTAS}, see also \cite[Theorem 5]{Yoshida2014SheraliAdams} for the exact statement used here.

\begin{theorem}
\label{thm:akk-dense-csp}
Fix \(k\), \(\abs{\Omega}\), and \(\delta>0\).  There is an algorithm which, given a dense
Max-\(k\)-CSP instance \(\mathcal I\) on \(n\) variables with alphabet size
\(\abs{\Omega}\) and predicates in \([0,1]\), outputs an assignment \(z\in\Omega^n\)
satisfying, with probability at least $\frac23$,
\[
 \val_{\mathcal I}(z)\ge \OPT(\mathcal I)-\delta.
\]
For fixed $k$,$\abs{\Omega}$, and $\delta$, the running time is polynomial in \(n\).
\end{theorem}
Notice that one can amplify the success probability of the algorithm by invoking $\ell$ independent runs, then return the assignment $z$ with largest $\val_{\mathcal I}(z)$. Such assignment is $\delta$-close to the optimum, except with probability $3^{-\ell}$.

\subsection{The Optimization Algorithm}\label{sec:opt}

Starting from a Hermitian matrix
\(M\in L(\C^d\otimes \C^d)\) with \(\|M\|_{\Frob}\le 1\), we construct a dense
Max-\(4\)-CSP whose optimum encodes the (penalized) product-state
optimization problem. Fix an accuracy parameter $\eta >0$ and an error parameter $0<\delta<\frac12$, the reduction proceeds as follows:
\begin{itemize}
\item[(i)] Decompose $M$, as in \Cref{eq:local-decomp}, into
    \[
    M=\alpha\Id\otimes\Id+A\otimes\Id+\Id\otimes B+C
    \]
    where $A,B$ are traceless and $\Tr_0 C = \Tr_1 C = 0$.
    \item[(ii)] Sample $U,V\sim \mathrm{Haar}(\mathcal{U}(d))$ and define
    \[
    \hat{M} := (U\otimes V)^*M(U\otimes V).
    \]
    Consequently, define $\hat{C} := (U\otimes V)^*C(U\otimes V)$.
\item[(iii)] Apply the decomposition induced by \Cref{lem:trunc-spread} with parameter $\delta$ to $\hat C = C_{\rm flat} + C_{\rm tail}$. If the decomposition fails, abort.
\item[(iv)] 

Define the Hilbert–Schmidt orthogonal projection onto the subspace with vanishing partial traces:
\[
\Pi(X) :=
X-\frac1d\Tr_1(X)\otimes \Id
-\Id\otimes\frac1d\Tr_0(X)
+\frac{\Tr X}{d^2}\Id\otimes \Id,
\]
and set $D:=\Pi(C_{\mathrm{flat}})$ and let $\gamma := d^2\norm*{D}_{\rm max}$.

\item[(v)] Apply \Cref{lem:dimension-free-discretization} with target accuracy \(\eta\), to obtain parameters $\kappa,\lambda,1/\varepsilon,\Delta=\poly(\gamma,1/\eta)$. Define $s:=\gamma \kappa^4
 +2\lambda(\kappa^2+1)^2+1$, and let \(\mathcal E\subset \kappa\mathbb D\) be an \(\varepsilon\)-net of the disk of radius $\kappa$.  
 
\item[(vi)] We now encode our optimization problem as a dense instance of Max-\(4\)-CSP.  The CSP has \(d\) variables $z_i\in\Omega:=\mathcal E\times\mathcal E$ for all $i\in[d]$, so an assignment consists of $d$-many pairs \(z_i=(x_i,y_i)\). For each ordered quadruple \((a,b,c,e)\in[d]^4\), define the predicate $ P_{a,b,c,e}:\Omega^4\to[0,1]$
by
\begin{align*}
&P_{a,b,c,e}(z_a,z_b,z_c,z_e)
:=
\\&\frac12+\frac{1}{2s}\left( 
\operatorname{Re}\!\left(
d^2 D_{ac,be}\,
{x}^*_ax_b
{y}^*_cy_e
\right)
-\lambda\left(|x_a|^2-1\right)
        \left(|x_b|^2-1\right) 
-\lambda\left(|y_c|^2-1\right)
 \left(|y_e|^2-1\right)\right).
\end{align*}
\end{itemize}
We show that the alphabet size of the CSP instance is bounded.
\begin{lemma}\label{lem:alphabet}
    $|\Omega| = \poly(1/\eta, \log(1/\delta))$.
\end{lemma}
\begin{proof}
  The alphabet size is
\[
 |\Omega|=|\mathcal E|^2
 \le
 c_{\rm net}^2\frac{\kappa^4}{\varepsilon^4}.
\]
Since both $\kappa$ and $1/\varepsilon$ are polynomials in $\gamma$ and $1/\eta$, it suffices to show that $\gamma = \poly(\log(1/\delta))$. By the triangle inequality and \Cref{lem:trunc-spread},
\[
\gamma = d^2\norm{D}_{\rm max} \leq 4d^2\norm{C_{\rm flat}}_{\rm max} \le c_{\rm flat}{4(1+\log(1/\delta))^2}\norm*{\hat C}_{\Frob}\leq c_{\rm flat}{4(1+\log(1/\delta))^2}
\]
using $\norm*{\hat C}_{\Frob} \leq 1$ since $\norm*{\hat M}_{\Frob} = \norm*{M}_{\Frob}\leq 1$.
\end{proof}
Next, we prove that each predicate is indeed bounded.
\begin{lemma}\label{lem:01}
    For any assignment $z \in \mathcal{E}^d\times \mathcal{E}^d$ and any quadruple $(a,b,c,e)\in[d]^4$, \[P_{a,b,c,e}(z_a,z_b,z_c,z_e) \in [0,1].\]
\end{lemma}
\begin{proof}
Since $\mathcal{E} \subset \kappa\mathbb D$, \[\abs{
 d^2D_{ac,be} x^*_a x_b y^*_c y_e
 }
 \le \gamma \kappa^4.\]
Moreover, $\abs{(|x_a|^2-1)(|x_b|^2-1)}
 \le (\kappa^2+1)^2$ and similarly for the $y$ term. 
Hence the absolute value of the expression inside the parentheses defining \(P_{a,b,c,e}\) is at most
\[
 \gamma \kappa^4+2\lambda(\kappa^2+1)^2
 \le s.
\]
Thus the predicate lies in $[0,1]$.
\end{proof}
We show that the value of the Max-$4$-CSP corresponds to the penalized objective function defined in \Cref{eq:penalized}.

\begin{lemma}\label{lem:lambda}
    For any assignment $z = (x,y) \in \mathcal{E}^d\times \mathcal{E}^d$, $\val_{\mathcal I}(z)
 = \frac{1}{2} + \frac{1}{2s}\Lambda_{D}(x,y)$.
\end{lemma}
\begin{proof}
    Expanding the definition of the value we have
\[
\val_{\mathcal I}(z)
 =
 \frac1{d^4}
 \sum_{a,b,c,e\in[d]}
 P_{a,b,c,e}(z_a,z_b, z_c, z_e)
 \]
and expanding each summand we can see that it indeed corresponds to a term in the $\Lambda_D(x,y)$ function. For the quartic term
\[
 \frac1{d^4}\sum_{a,b,c,e}\operatorname{Re}\!\left(
d^2 D_{ac,be}\,
x^*_ax_b
y^*_cy_e
\right) =  \frac1{d^2}\sum_{a,b,c,e}\operatorname{Re}\!\left(
D_{ac,be}\,
x^*_ax_b
y^*_cy_e
\right) = Q_{D}(x,y)
\]
since $D$ is Hermitian because $C_{\rm flat}$ is and the projection $\Pi$ preserves this property.
Also,
\[
 \frac1{d^4}\sum_{a,b,c,e}
 (|x_a|^2-1)(|x_b|^2-1)
 =
 \frac1{d^2}\sum_{a,b}
 (|x_a|^2-1)(|x_b|^2-1)
 =
 \left(\frac1d\sum_a |x_a|^2-1\right)^2,
\]
and similarly for the \(y\)-penalty term.
\end{proof}
We show that the optimum of the Max-$4$-CSP is close to the separable support function.
\begin{lemma}\label{lem:hSep-bound} If $d \geq \Delta$, then
$\abs{
h_{\Sep}(M)
-
\left(\alpha+2s\left(\OPT(\mathcal I)-\frac12\right)\right)
}
\le \delta+\eta + \frac{2}{\sqrt{d}}$.
\end{lemma}
\begin{proof}
First, local unitaries preserve the separable support function, since \(x,y\) range over all unit vectors if and only if \(Ux,Vy\) do. Hence $h_{\Sep}(M)=h_{\Sep}(\hat M)$.
For every product unit vector \(x\otimes y\),
\begin{align}
\abs{(x\otimes y)^*(\hat M-\alpha \Id \otimes \Id - \hat C)(x\otimes y)} &=
\abs{(x\otimes y)^*
  (\hat A\otimes \Id+\Id\otimes \hat B)
  (x\otimes y)}
=\abs{x^*\hat Ax+y^*\hat By}\notag\\
&\le \|\hat A\|_{\mathrm{op}}+\|\hat B\|_{\mathrm{op}}\le \|A\|_{\mathrm F}+\|B\|_{\mathrm F}\le \frac{2}{\sqrt d}\label{eq:hatM}
\end{align}
where $\hat A = U^*AU$ and $\hat B = V^*BV$. Taking suprema gives $\abs*{h_{\Sep}(\hat M)-\alpha - h_{\Sep}(\hat C)} \le \frac{2}{\sqrt d}$. Similarly, for all units $x\otimes y$,
\begin{equation}\label{eq:C-D}
\abs{(x\otimes y)^*(\hat C - D)(x\otimes y)} = \abs{(x\otimes y)^*\Pi(C_{\rm tail})(x\otimes y)} \leq \norm{C_{\rm tail}}_{\Frob} \leq \delta 
\end{equation}
by \Cref{lem:trunc-spread} and using the fact that $\Pi$ is contractive and \[\hat C - D = \hat C - \Pi(C_{\rm flat}) = \Pi(\hat C - C_{\rm flat}) = \Pi(C_{\rm tail}).\] Therefore $\abs*{h_{\Sep}(\hat C)-h_{\Sep}( D)} \le \delta$. Using \Cref{eq:quartic}, we rewrite
\[
h_{\Sep}( D) = \max_{\norm{x}_2=\norm{y}_2=1} (x \ot y)^* D (x\ot y) = \max_{\norm{x}_2=\norm{y}_2=\sqrt d}
 Q_{ D}(x,y).
\]
By \Cref{lem:dimension-free-discretization}, for $d\geq \Delta$,
\[
\left|
 \max_{\norm{x}_2=\norm{y}_2=\sqrt d}
 Q_{ D}(x,y)
 -
 \max_{x,y\in\mathcal E^d}\Lambda_{D}(x,y)
\right|
\le \eta.
\]
Finally, by \Cref{lem:lambda},
\[
 \max_{x,y\in\mathcal E^d}\Lambda_{D}(x,y)
 =
 2s\left(\OPT(\mathcal I)-\frac12\right).
\]
Combining the previous estimates with a triangle inequality gives the desired bound.
\end{proof}
Finally, we are ready to state our main theorem.
\begin{theorem}\label{thm:final}
For any constant $0<\tau \leq1$, there exists a randomized algorithm that, on input a Hermitian matrix $M\in L(\mathbb C^d \ot \mathbb C^d)$ with $\norm{M}_{\Frob}\leq 1$, returns a pair of unit vectors $x,y$ such that
    \[
    (x\ot y)^* M (x\ot y)\geq h_{\Sep}(M) -\tau
    \]
    with probability at least $2/3$.
    For a fixed $\tau$, the runtime is polynomial in $d$.
\end{theorem}
\begin{proof}
Fix some $0<\tau \leq1$. Set $\delta:=\frac{\tau}{16}$,  $\eta:=\frac{\tau}{4}$, and $\ell$ such that $3^{-\ell} + \delta^2 \leq \frac13$. We assume without loss of generality that $d\ge \max\{\Delta,\frac{1024}{\tau^2}\}$, as otherwise the dimension is constant, and therefore the problem can be solved using known methods, see e.g., \cite{Brandao2011Quasipolynomial}. 

The algorithm proceeds as follows:
\begin{itemize}
    \item[(i)] Run the reduction as specified above with parameters $\eta$ and $\delta$, and let $\mathcal{I}$ be the resulting dense Max-$4$-CSP instance. If the reduction aborts, return an arbitrary pair of unit vectors.
    \item[(ii)] Run $\ell$ independent instances of the algorithm from \Cref{thm:akk-dense-csp} with accuracy parameter $\frac{\tau}{8s}$ and let $z = (x,y)$ be the assignment with the highest value.
    \item[(iii)] Define the \(\sqrt d\)-normalized vectors
\[
 \tilde{x}:=
 \begin{cases}
 \sqrt d\,x/\norm{x}_2, & x\neq 0,\\
 \sqrt d\,u_1, & x=0,
 \end{cases}
 \quad\text{and}\quad
 \tilde{y}:=
 \begin{cases}
 \sqrt d\,y/\norm{y}_2, & y\neq 0,\\
 \sqrt d\,u_1, & y=0.
 \end{cases}
\]
\item[(iv)] The algorithm returns the unit vectors
\[
 \hat x:=\frac{U\tilde{x}}{\sqrt d}
 \quad\text{and}\quad
 \hat y:=\frac{V\tilde{y}}{\sqrt d}.
\]
\end{itemize}
First, notice that the probability that the flat/tail decomposition fails is bounded by $\delta^2$, by \Cref{lem:trunc-spread}. Moreover, the probability that the $z$ returned by the algorithm from \Cref{thm:akk-dense-csp} is not $\frac{\tau}{8s}$-close to the optimum is bounded by $3^{-\ell}$. Thus, with probability at least $\frac23$, the reduction does not abort and the guarantee of \Cref{thm:akk-dense-csp} applies, and we henceforth assume that this is the case. \Cref{lem:01} and \Cref{lem:alphabet} establish that the output of the reduction is a valid dense Max-\(4\)-CSP, with alphabet size
\[
 \abs{\Omega}=\poly(1/\eta,\log(1/\delta))=\poly(1/\tau).
\]
Since $s=\poly(1/\tau)$ as well, the overall runtime is polynomial in $d$. 

What is left to be shown is that the output of the algorithm indeed approximates $h_{\Sep}(M)$. First notice that
\[
\begin{aligned}
 \alpha + Q_{D}(\tilde{x},\tilde{y})
 &\ge \alpha + 
 \Lambda_{D}(x,y)
 -
 \frac{\gamma^2+\gamma^2 \kappa^4}{4\lambda}\\
 &\ge \alpha 
 + 2s\left(\OPT(\mathcal I)-\frac12\right)-\frac{\tau}{4} - \frac{\gamma^2+\gamma^2 \kappa^4}{4\lambda}\\
 &\ge h_{\Sep}(M)-\delta - \frac{2}{\sqrt{d}}-\eta-\frac{\tau}{4} - \frac{\gamma^2+\gamma^2 \kappa^4}{4\lambda}
\end{aligned}
\]
where the first inequality is \Cref{claim:1}, the second inequality is by \Cref{lem:lambda} along with the fact that $\val_{\mathcal I}(z)\ge \OPT(\mathcal I)-\frac{\tau}{8s}$ by assumption, and the third inequality is by \Cref{lem:hSep-bound}.
This implies that 
\begin{align*}
    (\hat x\ot \hat y)^*M(\hat x \ot \hat y)
 &=
 (\tilde{x}/\sqrt d\ot\tilde{y}/\sqrt d)^*\hat M
 (\tilde{x}/\sqrt d\ot\tilde{y}/\sqrt d) \\
 &\ge
 (\tilde{x}/\sqrt d\ot\tilde{y}/\sqrt d)^*(\alpha \Id \otimes \Id + D)
 (\tilde{x}/\sqrt d\ot\tilde{y}/\sqrt d)
 -\delta - \frac{2}{\sqrt{d}} \\
 &= \alpha+
 Q_{D}(\tilde{x},\tilde{y}) -\delta - \frac{2}{\sqrt{d}}\\
 &\ge
 h_{\Sep}(M)
 -2\delta-\eta-\frac{\tau}{4}-\frac{4}{\sqrt{d}}
 -\frac{\gamma^2+\gamma^2 \kappa^4}{4\lambda}
\end{align*}
where the second line uses \Cref{eq:hatM} and \Cref{eq:C-D}. With our choice of $\lambda$ in \Cref{lem:dimension-free-discretization}, we obtain
\[
(\hat x \ot \hat y)^*M(\hat x \ot \hat y) \geq h_{\Sep}(M) - 2\delta-\eta-\frac{\tau}{4}-\frac{4}{\sqrt{d}}-\eta \ge 
h_{\Sep}(M) - \tau
\]
as desired.
\end{proof}

\subsection{From Optimization to Weak Membership}\label{sec:wmemb}
\label{sec:opt-to-membership}

We show how to use the algorithm from \Cref{thm:final} to construct a solver for the weak membership problem for $\Sep(d,d)$. This follows as a standard application and analysis of the Frank-Wolfe algorithm and more generally as a consequence of existing works, such as \cite{Jaggi2013FrankWolfe}. We report the proof here for the sake of completeness.

The algorithm takes as input a Hermitian $\rho \in L(\C^d \ot \C^d)$ with $\norm{\rho}_{\Frob} \leq 1$ and an accuracy parameter $0<\eta\le 1$.
Start from an arbitrary separable state, e.g., $\sigma_1 := (u_1\ot u_1)(u_1\ot u_1)^*$ and set $T:=\frac{256}{\eta^2}$ and $\ell$ such that $T3^{-\ell}\leq \frac13$. Then for $t = 1,\dots, T-1$ proceed as follows:
\begin{itemize}
    \item[(i)] If $\sigma_t-\rho = 0$, accept.
    \item[(ii)] Otherwise run $\ell$ independent instances the algorithm from \Cref{thm:final} with accuracy parameter $\tau:= \frac{\eta^2}{64}$ on input $\frac{\rho-\sigma_t}{\norm{\rho-\sigma_t}_{\Frob}}$. Define $\omega_t\in\Sep(d,d)$ to be the output of the algorithm with the observed objective value. 
    \item[(iii)] Update $\sigma_{t+1} := \left(1-\frac{2}{t+2}\right)\sigma_t +\frac{2}{t+2}\omega_t$.
\end{itemize}
Accept if $\norm{\rho-\sigma_T}_{\Frob}\le \eta/2$ and reject otherwise.

By a union bound, the probability that any state $\omega_t$ does not satisfy the condition from \Cref{thm:final} is bounded by $T3^{-\ell}$, so we henceforth assume that this is the case. Moreover, each $\sigma_t \in\Sep(d,d)$ since $\Sep(d,d)$ is a convex set. To prove convergence, define the potential function
\[
f_\rho(\sigma) := \frac12 \norm{\sigma -\rho}^2_{\Frob} -\min_{\psi\in\Sep(d,d)}\frac12 \norm{\psi -\rho}^2_{\Frob}.
\]
We prove the following recursive relation.
\begin{lemma}\label{lem:recursion}
    $ f_\rho(\sigma_{t+1})-2\tau
 \le
 \left(1-\frac{2}{t+2}\right)(f_\rho(\sigma_{t})-2\tau)+2\left(\frac{2}{t+2}\right)^2$, for all $1\leq t\leq T-1$.
\end{lemma}
\begin{proof}
    By \Cref{thm:final}, the output of the optimization algorithm satisfies
    \[
    \Tr\left(\frac{\rho-\sigma_t}{\norm{\rho-\sigma_t}_{\Frob}} \omega_t\right) \geq h_{\Sep}\left(\frac{\rho-\sigma_t}{\norm{\rho-\sigma_t}_{\Frob}}\right)-\tau.
    \]
Equivalently,
\[
 \ip{\sigma_t-\rho}{\omega_t}
 \le
 \min_{\psi\in\Sep(d,d)}\ip{\sigma_t-\rho}{\psi}
 + \tau \norm{\sigma_t-\rho}_{\Frob} \le \min_{\psi\in\Sep(d,d)}\ip{\sigma_t-\rho}{\psi}
 + 2\tau 
\]
since $\norm{\sigma_t-\rho}_{\Frob} \leq 2$. By linearity and convexity,
\[
\ip{\sigma_t-\rho}{\omega_t -\sigma_t}
 \le \min_{\psi\in\Sep(d,d)}\ip{\sigma_t-\rho}{\psi-\sigma_t}
 + 2\tau \leq -f_\rho(\sigma_t) +2\tau.
\]
Using the identity $\norm{A+B}_{\Frob}^2
=
\norm{A}_{\Frob}^2+2\ip{A}{B}+\norm{B}_{\Frob}^2$, we can expand
\begin{align*}
\frac12\norm{\sigma_{t+1}-\rho}_{\Frob}^2 &=\frac12\norm{\sigma_t+\frac{2}{t+2}(\omega_t-\sigma_t)-\rho}_{\Frob}^2\\
&=
\frac12\norm{\sigma_t-\rho}_{\Frob}^2
+\frac{2}{t+2}\ip{\sigma_t-\rho}{\omega_t-\sigma_t}
+\frac{1}{2}\left(\frac{2}{t+2}\right)^2\norm{\omega_t-\sigma_t}_{\Frob}^2\\
&\le \frac12\norm{\sigma_t-\rho}_{\Frob}^2
+\frac{2}{t+2}\ip{\sigma_t-\rho}{\omega_t-\sigma_t}+ 2\left(\frac{2}{t+2}\right)^2.
\end{align*}
Combining the two inequalities we obtain
\[
f_\rho(\sigma_{t+1}) \leq \left(1- \frac{2}{t+2}\right)f_\rho(\sigma_t) + \frac{4\tau}{t+2} + 2\left(\frac{2}{t+2}\right)^2
\]
and the lemma statement follows by subtracting $2\tau$ from both sides.
\end{proof}
We claim that $f_\rho(\sigma_{t})-2\tau \le \frac{8}{t+2}$ and we show it with an induction. The base case is immediate, whereas for the inductive step
\[
f_\rho(\sigma_{t+1})-2\tau\le \frac{t}{t+2}(f_\rho(\sigma_{t})-2\tau)+\frac{8}{(t+2)^2}
\le \frac{t}{t+2}\cdot\frac{8}{t+2}+\frac{8}{(t+2)^2}
=\frac{8(t+1)}{(t+2)^2}
\le \frac{8}{(t+1)+2}
\]
by \Cref{lem:recursion}. For $T=\frac{256}{\eta^2}$, we obtain that
\[
f_\rho(\sigma_{T})\le 2\tau + \frac{8}{T+2} \leq \frac{\eta^2}{32} +\frac{\eta^2}{32} =\frac{\eta^2}{16}.
\]
Now observe that if $\rho\in\Sep(d,d)$, then $\min_{\psi\in\Sep(d,d)}\frac12 \norm{\psi -\rho}^2_{\Frob} = 0$, therefore
\[
 \norm{\rho-\sigma_T}_{\Frob}
 =
 \sqrt{2f_\rho(\sigma_T)}
 \le
 \frac{\eta}{2}
\]
so the algorithm correctly accepts after $T$ iterations. On the other hand, if $\min_{\psi\in\Sep(d,d)} \norm{\rho-\psi}_{\Frob} > \eta$ then the algorithm always rejects since $\sigma_T \in \Sep(d,d)$. Thus we have the following implication.

\begin{corollary}
\label{cor:constant-gap-sepmem}
For any fixed $0 < \eta \le 1$, there is a randomized algorithm that solves the weak membership problem for \(\Sep(d,d)\) with gap \(\eta\), in time polynomial in $d$.
\end{corollary}

\subsection*{Acknowledgements} 
The author wishes to thank Michael Walter for guidance, for suggesting the proof that a weak membership algorithm for $\Sep(d,d)$ implies one for $\Sep(d_0,d_1)$, and for general feedback on this manuscript.

Work supported by the European Research Council through an ERC Starting Grant (Grant agreement No.~101077455, ObfusQation) and by the Deutsche Forschungsgemeinschaft (DFG, German Research Foundation) under Germany's Excellence Strategy - EXC 2092 CASA – 390781972.

\subsection*{Statement on AI Use}
The main idea for this proof was generated biologically, but ChatGPT 5.5 and 5.6 assisted on the proofs of all technical statements involved in the reduction. The final version of this manuscript is written by the author, who takes full responsibility for its content.

\bibliographystyle{alpha}
\bibliography{references}

\appendix
\section{Bit Complexity}\label{sec:bit}
We sketch how to prove the same theorem in the standard finite-precision Turing model of computation. Assume that \(M\) has rational entries. Run the proof of
\Cref{thm:final} with accuracy \(\tau/2\), choosing \(\ell\) so that its
good events hold with probability at least \(5/6\), and set
$r:=\min\{d^{-2},\tau/100\}$.

By, e.g., \cite[Proposition C.9]{banks2023pseudospectral}, using \(O(\log d+\log(1/r))\) bits of precision, produces
\(\tilde U\) which can be coupled to a Haar unitary \(U\) so that
\[
 \norm{\tilde U-U}_{\mathrm{op}}\le r/100
\]
except with probability \(1/12\). Apply this independently to obtain
\(\tilde U,\tilde V\), coupled to independent Haar unitaries \(U,V\),
and define
\[
 \hat C=(U\otimes V)^*C(U\otimes V)\quad \text{and} \quad
 \tilde C=\Pi\!\left((\tilde U\otimes\tilde V)^*
 C(\tilde U\otimes\tilde V)\right).
\]
Since \(\norm{C}_{\Frob}\le1\) and \(\Pi\) is contractive, expanding the
conjugations gives $\norm*{\tilde C-\hat C}_{\Frob}\le r$.

In step~(iii), threshold \(\tilde C\) at the cutoff in the proof of
\Cref{lem:trunc-spread}, increased by \(r\), and abort only if the new tail
has Frobenius norm greater than \(\delta+r\). On the good Haar event, every
entry in the new tail belongs to the original tail, and hence
\[
 \norm{\tilde C_{\rm tail}}_{\Frob}\le\delta+r,\qquad
 \norm{\tilde C_{\rm flat}}_{\rm max}
 \le c_{\rm flat}\frac{(1+\log(1/\delta))^2}{d^2}+r.
\]
Thus, for \(D:=\Pi(\tilde C_{\rm flat})\),
\[
 d^2\norm{D}_{\rm max}
 \le4c_{\rm flat}(1+\log(1/\delta))^2+4,
 \qquad
 \norm{\hat C-D}_{\Frob} \leq \norm{\tilde C-\hat C}_{\Frob}+ \norm{\tilde C_{\rm tail}}_{\Frob}\le\delta+2r.
\]
The discretization parameters therefore remain \(\poly(1/\tau)\), while
\[
 \abs{h_{\Sep}(\hat C)-h_{\Sep}(D)}
 \le\norm{\hat C-D}_{\Frob}\le\delta+2r.
\]
Consequently, the two uses of the bound \(\delta\) in the proof of
\Cref{thm:final} cost at most \(4r\) more.

Finally, return
\[
 \hat x=\frac{\tilde U\tilde x}{\norm{\tilde U\tilde x}_2},
 \qquad
 \hat y=\frac{\tilde V\tilde y}{\norm{\tilde V\tilde y}_2}.
\]
Each vector is within \(r/50\) of its exactly rotated counterpart, so
\(\norm{M}_{\mathrm{op}}\le1\) implies an additional objective error smaller
than \(r\). The total error is therefore at most
\(\tau/2+5r<\tau\), and the success probability is at least
\(5/6-2(1/12)=2/3\).

Regarding the CSP, use rational upper
bounds for all scalar parameters and a rational net \(\mathcal E\), and
perform the remaining arithmetic exactly. Since
\(\log(1/r)=O(\log d+\log(1/\tau))\), this has polynomial bit complexity.

\end{document}